\documentclass[runningheads]{llncs}
\usepackage{amssymb}
\usepackage[usenames,dvipsnames,svgnames,table]{xcolor}
\usepackage{amsmath}
\usepackage{graphicx}
\usepackage{hyperref}
\usepackage{color}
\usepackage{epstopdf}
\usepackage{cleveref}
\usepackage[svgnames]{xcolor}
\usepackage{braket}
\usepackage{cite}
\hypersetup{hidelinks,colorlinks=true,allcolors=DarkBlue}

\usepackage[normalem]{ulem}

\usepackage{enumitem}
\usepackage[ruled, vlined]{algorithm2e}

\usepackage[normalem]{ulem}

\newcommand{\Adv}{{\rm Adv}}
\newcommand{\Av}{{\mathcal{A}}}
\newtheorem{defn}{Definition}

\newcommand{\address}{\ensuremath{{\rm address}}}

\newcommand{\mk}{\ensuremath{{\rm mk}}}

\newcommand{\RO}{\ensuremath{{\rm RO}}}

\begin{document}
\title{Advanced attribute-based protocol based on the modified secret sharing scheme}
\author
{M.A. Kudinov$^{1,2,3}$\and
A.A.~Chilikov$^{3,4}$\and
E.O. Kiktenko$^{1,2}$\and
A.K. Fedorov$^{1,2}$
}

\authorrunning{M.A. Kudinov et al.}
\titlerunning{Advanced attribute-based protocol}
\institute{$^1$Russian Quantum Center, Skolkovo, Moscow 143025, Russia\\ $^2$ QApp, Skolkovo, Moscow 143025, Russia \\ $^3$Bauman Moscow State Technical University, Moscow 105005, Russia \\ $^4$ Moscow Institute of Physics and Technology, Dolgoprudny, Moscow Region 141700, Russia \\
\email{mishel-kudinov@mail.ru; chilikov@passware.com; e.kiktenko@rqc.ru; akf@rqc.ru}}

\maketitle
\begin{abstract}
We construct a new protocol for attribute-based encryption with the use of the modification of the standard secret sharing scheme. 
In the suggested modification of the secret sharing scheme, only one master key for each user is required that is achieved by linearly enlarging public parameters in the access formula.
We then use this scheme for designing an attribute-based encryption protocol related to some access structure in terms of attributes. 
We demonstrate that the universe of possible attributes does not affect the resulting efficiency of the scheme.
The security proofs for both constructions are provided.
\keywords{secret sharing \and attribute-based encryption \and monotone access structures}
\end{abstract}

\begin{abstract}
We construct a new protocol for attribute-based encryption with the use of the modification of the standard secret sharing scheme. 
In the suggested modification of the secret sharing scheme, only one master key for each user is required that is achieved by linearly enlarging public parameters in access formula.
We then use this scheme for designing an attribute-based encryption protocol related to some access structure in terms of attributes. 
We demonstrate that the universe of possible attributes does not affect the resulting efficiency of the scheme.
The security proofs for both constructions are provided.
\keywords{Secret sharing \and Attribute-based encryption \and Monotone access structures}

\end{abstract}

\section{Introduction} \label{intro}
In the view of the significant increase in the amount of digital communications, the problem of efficient protection of data becomes crucial.
An important task is to construct a secured protocol for controlled access to data.
In standard protocols for solving this problem, which are mostly based on public-key cryptography, a secret key is required for access to whole encrypted data. 
A straightforward modifications of such protocols for providing partial access to data lead to a significant increase of the complexity since multiple encryptions of the same data are needed. 

Attribute-based encryption (ABE) is a relatively new approach for solving the data access control problem~\cite{Goyal2006,Boneh2001,Boneh2005}.
In the ABE schemes, the access to the parts of an encrypted data is determined by a set of \emph{attributes}, which are inherent to various participants.
Thus, if attributes of a participant belonging to a particular subset of possible attributes, then he is able to obtain access to a corresponding particular part of the encrypted data.
The ABE conception appears to be very promising in a framework of cloud technologies and distributed ledgers.
Over the past decade, a number of modifications and improvements have been presented~\cite{Goyal2006,Sahai2005,Abdalla2006}.
However, some of the proposed approaches still suffer from implementation
complexity, which increases with the number of attributes.

We note that the concept of ABE has much in common with the secret sharing (SS) problem. 
However, one of the most common SS schemes~\cite{Benaloh1990} 
has a problem related to a large number of shares per trustee. 

In this work, we propose an advanced ABE protocol with a sufficiently low computational complexity.
One of the main techniques of our work is a modification of the standard SS scheme, which allows one to use a single key for generating the whole set of required shares. 
This modification is then used for the construction of the ABE protocol, which is independent to the size of the set of possible attributes.

The paper is organized as follows.
In Sec.~\ref{sec:prel} we provide basic definitions.
In Sec.~\ref{sec:stand} we briefly describe the standard construction of the general SS scheme.
In Sec.~\ref{sec:advss} we present our modification of the SS scheme, which is then used for constructing the ABE protocol in Sec.~\ref{sec:advabe}.
In Sec.~\ref{sec:efficiency} we estimate the required resources for encryption and decryption for the suggested protocol. 
We summarize our results and conclude in Sec.~\ref{sec:conc}.

\section{Preliminaries} \label{sec:prel}
Let us introduce basic definitions and notations.

Let $x \leftarrow \mathcal{X}$, where $x$ is a random value and $\mathcal{X}$ is a probability distribution, denote a sampling of $x$ from the distribution $\mathcal{X}$. 
Let $y \leftarrow M(x)$, where $M$ is an algorithm, denote the output $y$ of $M$ processed on the input $x$.
Let $x \stackrel{\$}{\leftarrow} X$, where $X$ is a set, denote an element $x$, which is chosen uniformly at random from the set $X$.
Let $\lor(\phi_1, \ldots, \phi_n)$ and $\land (\phi_1, \ldots, \phi_n)$ stand for $\phi_1 \lor \ldots \lor \phi_n$ and $\phi_1 \land \ldots \land \phi_n$, correspondingly.

\newcommand{\rand}{\stackrel{\$}{\leftarrow}}

Now we define a pseudorandom function (PRF) family.  
Given the oracle $f$, we denote $M(f)$ as the execution of the oracle machine $M$ with an access to $f$.

\begin{defn}[pseudorandom function (PRF) family]
    We define $\mathbb{F}_{\mathcal{D}\rightarrow\mathcal{E}}=\{f_k: \mathcal{D} \rightarrow \mathcal{E}\}_{k \in \mathcal{K}}$, where $|\mathcal{K}| = |\mathcal{D}| = |\mathcal{E}| < \infty$ 
    to be a function family.
    We define the advantage of an adversary $\mathcal{A}$ against PRF as
    \begin{equation*}   	
    		{\rm Adv_{\mathbb{F}_{\mathcal{D}\rightarrow\mathcal{E}}}^{PRF}(\mathcal{A})} = |\Pr[1 \leftarrow \mathcal{A}(f_k): k \stackrel{\$}{\leftarrow} \mathcal{K}] - \Pr[1 \leftarrow \mathcal{A}(h): h \stackrel{\$}{\leftarrow} H_{\mathcal{D}\rightarrow \mathcal{E}}]| ,
    \end{equation*}
    where $H_{\mathcal{D}\rightarrow\mathcal{E}} $ is a family of all functions from $\mathcal{D} \rightarrow \mathcal{E}$ ($|H_{\mathcal{D}\rightarrow\mathcal{E}}|=|\mathcal{E}|^{|\mathcal{D}|}$).
    We define the PRF insecurity of a function family $\mathbb{F}_{\mathcal{D}\rightarrow\mathcal{E}}$
    against time-$\xi$ adversaries as the maximum advantage of any
classical adversary that runs in time $\xi$ :
\begin{equation*}
    {\rm InSec^{PRF}(\mathbb{F}_{\mathcal{D}\rightarrow\mathcal{E}}, \xi)
     = \underset{\mathcal{A}}{max} \{Adv_{\mathbb{F}_{\mathcal{D}\rightarrow\mathcal{E}}}^{PRF}(\mathcal{A})\}}
\end{equation*}
\end{defn}

\begin{defn}[$m$-PRF family game]
    We say that an oracle $\omega$ is initialized with a function $f(\cdot)$ if $\omega(x)=f(x)$, and denote it as $\omega \leftarrow f$.
    The following procedure is called $m$-PRF family game
    \begin{description}
        \item[Init:] Given $\mathbb{F}_{\mathcal{D}\rightarrow\mathcal{E}}=\{f_k: \mathcal{D} \rightarrow \mathcal{E}\}_{k \in \mathcal{K}}$,
        where $|\mathcal{K}| = |\mathcal{D}| = |\mathcal{E}|$, flip a fair coin $b$.
        If $b=1$ then $\Omega=\{\omega_i \leftarrow f_k : k\stackrel{\$}{\leftarrow}\mathcal{K}, \ i\in \{1, \ldots, m\} \}$. 
        Otherwise $\Omega=\{\omega_i \leftarrow h : h\stackrel{\$}{\leftarrow}H_{\mathcal{D}\rightarrow\mathcal{E}}, \ i\in \{1, \ldots, m\}\}$, where $H_{\mathcal{D}\rightarrow\mathcal{E}} $ is a family of all functions from $\mathcal{D} \rightarrow \mathcal{E}$.
        \item[Game:] Given a set of oracles $\Omega$, the challenge is to distinguish whether $\Omega$ is initialized with functions from $F_{\mathcal{D}\rightarrow\mathcal{E}}$ or from $H_{\mathcal{D}\rightarrow\mathcal{E}}$
    \end{description}
    We define the advantage of an adversary $\mathcal{A}$ against $m$-PRF as
    \begin{equation*}
        {\rm Adv}_{\mathbb{F}_{\mathcal{D}\rightarrow\mathcal{E}}}^{m-{\rm PRF}}(\mathcal{A}) = |\Pr[1\leftarrow \mathcal{A}(\Omega)| b=1] - \Pr[1 \leftarrow \mathcal{A}(\Omega) | b=0]|.
    \end{equation*}
    We define the $m$-PRF insecurity of a function family $\mathbb{F}_{\mathcal{D}\rightarrow\mathcal{E}}$
    against time-$\xi$ adversaries as the maximum advantage of any
classical adversary that runs in time $\xi$ :
\begin{equation*}
    {\rm InSec}^{m-{\rm PRF}}
    (\mathbb{F}_{\mathcal{D}\rightarrow\mathcal{E}}, \xi)
     = \underset{\mathcal{A}}{\max} \{{\rm Adv}_{\mathbb{F}_{\mathcal{D}\rightarrow\mathcal{E}}}^{m-{\rm PRF}}(\mathcal{A})\}
\end{equation*}
\end{defn}

\begin{defn}[Decisional Diffie–Hellman (DDH) challenge~\cite{Boneh1998,Cramer2012}]
    Consider a (multiplicative) cyclic group $G$ of the order $q$ with the generator $g$.
    We define the advantage of an adversary $\mathcal{A}$ against DDH as
    \begin{equation}
       {\rm Adv}_{G}^{\rm DDH}(\mathcal{A}) = |\Pr(1\leftarrow\mathcal{A}(g^a, g^b, g^{ab}) - \Pr(1 \leftarrow \mathcal{A}(g^a, g^b, g^z))| 
    \end{equation}
    where $a,b,z$ are chosen randomly and independently from $\mathbb{Z}_q$.
    We define the DDH insecurity of a group $G$
    against time-$\xi$ adversaries as the maximum advantage of any
classical adversary that runs in time $\xi$ :
\begin{equation*}
    {\rm InSec^{DDH}}(G, \xi) = {\rm \underset{\mathcal{A}}{max}}
     \{{\rm Adv}_{G}^{\rm DDH}(\mathcal{A})\}
\end{equation*}
\end{defn}

\begin{defn}[$m$-DDH challenge]
    Consider a (multiplicative) cyclic group $G$ of the order $q$ with the generator $g$, and following two distibutions:    
    \begin{itemize}
        \item $\Omega_{ab}=\{(g^a, g^{b_1}, g^{a\cdot b_1}), (g^a, g^{b_2}, g^{a\cdot b_2}), \ldots ,(g^a, g^{b_m}, g^{a\cdot b_m})\} $, where $a$ and $b_i$ are chosen randomly and independently from $\mathbb{Z}_q$ for $i = 1, \ldots, m$,
        \item $\Omega_{z} =\{(g^a, g^{b_1}, g^{z_1}), (g^a, g^{b_2}, g^{z_2}), \ldots , (g^a, g^{b_m}, g^{z_m})\}$, where $a, b_i , z_i$ are chosen randomly and independently from $\mathbb{Z}_q$ for $i = 1, \ldots, m$,
    \end{itemize}
    We define the advantage of an adversary $\mathcal{A}$ against $m$-DDH as
    \begin{equation}
        {\rm Adv}_{G}^{ m-{\rm DDH}}(\mathcal{A}) = |\Pr[1\leftarrow\mathcal{A}(\Omega_{ab})] - \Pr[1 \leftarrow \mathcal{A}(\Omega_{z})]| 
     \end{equation}
     We define the DDH insecurity of a group $G$
    against time-$\xi$ adversaries as the maximum advantage of any
classical adversary that runs in time $\xi$ :
\begin{equation*}
    {\rm InSec}^{m-{\rm DDH}}(G, \xi)
     ={\rm \underset{\mathcal{A}}{max}} \{{\rm Adv}_{G}^{ m-{\rm DDH}}(\mathcal{A})\}
\end{equation*}

    \end{defn}

\begin{defn}Let $\mathcal{P} = \{P_1, \ldots P_n\}$ be a set. 
    An access structure $\mathcal{B}$ is a collection of non-empty subsets of $\mathcal{P}$, i.e., $\mathcal{B} \subseteq 2^{\mathcal P}$.
\end{defn}

\begin{defn}
    Given a set $\mathcal{P}$, a monotone access structure on $\mathcal{P}$ is a collection of subsets $\mathcal{B} \subseteq 2^{\mathcal P}$
    such that
    \begin{equation*}
        B \in \mathcal{B}, B \subseteq B' \subseteq \mathcal{P} \Rightarrow B' \in \mathcal{B}.
    \end{equation*}
\end{defn}

\begin{defn}
    A Boolean function $\Phi:\{0,1\}^n\rightarrow \{0,1\}$ is called monotone, if $\Phi(x_1, \ldots, x_n) \leq \Phi(x_1', \ldots, x_n')$, whenever for every $i\in\{1,\ldots,n\}$ $x_i \leq x_i'$.
\end{defn}

\begin{defn}
    Given an access structure $\mathcal{B}$, define a Boolean function $\Phi_\mathcal{B}: \{0,1\}^{|P|} \rightarrow \{0,1\}$ on $|\mathcal{P}|$-bit strings, where each bit is indexed by an element from $\mathcal{P}$,
    such that $\Phi(x)=1$ iff $\{p: x_p=1\} \in \mathcal{B}$. 
\end{defn}

One can look at the Boolean function $\Phi_\mathcal{B}$ as an indicator of the set $\mathcal{B}$.
It is easy to check that the defined $\Phi_{\mathcal{B}}$ is a monotone Boolean function for a proper monotone access structure $\mathcal{B}$. 

\begin{defn}
    For a given set $\mathcal{P}$ and a monotone access structure $\mathcal{B}$ on $\mathcal{P}$,
    define $\mathcal{F(B)}$ to be the set of all Boolean formulae
    (expressions consisted of logical operations) on $|P|$ variables,
    such that for every formula $\phi \in \mathcal{F(B)}$ the output of $\phi$
    is true iff the true variables in $\phi$ correspond exactly
    to a set $B \in \mathcal{B}$ 
    (here we assume that each Boolean variable in the formula is indexed with an element form $\mathcal{P}$).
\end{defn}

We note that $\phi$, $\phi' \in \mathcal{F(B)}$ implies that $\phi$ and $\phi'$ correspond to the same function $\Phi_{\mathcal B}$. 
They may, however, represent entirely different formula to express this function.

\begin{defn}[Random oracle~\cite{Bellare1993}]
    Random oracle is an oracle (a theoretical black box) that responds to every unique query with a value chosen uniformly at random from its output domain. 
    If a query is repeated, it responds the same way every time that query is submitted.
    We refer a set of independent Random Oracles, $\{{\rm RO}_1, \ldots, {\rm RO}_t\}$, as a family of Random Oracles.
\end{defn}

\section{Standard construction of the general SS scheme} \label{sec:stand}

We begin our consideration with a SS scheme, which is proposed by  J. Benaloh and J. Leichter~\cite{Benaloh1990}, that we refer to as a \emph{standard SS scheme}.
For this purpose we first introduce a definition of the \textit{secure generalized SS scheme}.
It is known that for certain access structures every secure generalized SS scheme must be able to assign multiple shares to each trustee (see Theorem~\ref{thm:2} below). 
In this case, we use $s_{p,j}$ to denote the $j^{\rm th}$ share given to trustee $p$.

We define the scheme with the use of the following roles.
We call the \emph{dealer}, a user who shares a secret according to some access structure. 
The \emph{trustees} are users among which the secret is shared. 
A \emph{party} is a group of trustees. 
We denote the set of all trustees as $\mathcal{P}$.

\begin{defn}[Secure generalized SS scheme]
    Given a monotone access structure $\mathcal{B}$ on a set of trusties
     $\mathcal{P}$ and a set of possible secrets $\mathcal{S}$,
      a secure generalized SS scheme for $\mathcal{B}$ is a method of dividing a secret $s\in\mathcal{S}$ into shares $\{s_{p,j}\}_{p\in\mathcal{P},j\in \mathbb{N}}$  such that
    \begin{itemize}
        \item for every $B \in \mathcal{B}$, there is an algorithm for reconstructing the secret $s$ from the subset of shares $\underset{p\in B}{\cup} \underset{j}{\cup} s_{p,j}$;
        \item for every $B \notin \mathcal{B}$ the subset of shares $\underset{p\in B}{\cup} \underset{j}{\cup} s_{p,j}$ provides no information (in an information theoretic sense) about the value of $s$.
    \end{itemize}
\end{defn}

In what is presented below, we define the secret domain $\mathcal{S}=\mathbb{Z}_q$, for some positive integer $q$. 
We then are able to construct the secure generalized SS scheme.

It is known that every monotone function $\Phi$ can represented with a formula $\phi$ consisted only of $\land $ and $\lor $ operations (without NOT operation).
It is then sufficient to demonstrate how to divide a secret ``across" these two operators. 
We use $X_{p,j}$ to denote the $j^{\rm th}$ appearance of variable $X_p: p \in P$ in a formula $\phi$. 
We refer it as $j$-notation. 
For example, a formula $(X_1 \land X_2)\lor(X_1 \land X_3)$ transforms to $(X_{1,1} \land X_{2,1})\lor (X_{1,2} \land X_{3,1})$

Let $\$(s, \phi)$ be a random function, which declares shares for each trustee $p \in P$ for $s \in \mathcal{S}$ and a monotone formula $\phi$, that is defined as follows
(we assume that $\phi$ is represented in $j$-notation):
\begin{itemize}
    \item $\$(s',X_{p,j})$ assigns the share $s'$ to trustee $p\in P$, such that $s_{p,j} := s'$;
    \item $\$(s',\lor(\phi_1, \ldots ,\phi_n)) =\underset{1\leq i \leq n}{\cup}\$(s,\phi_i)$;
    \item $\$(s,\land(\phi_1, \ldots ,\phi_n)) = \underset{1\leq i \leq n}{\cup}\$(s_i,\phi_i)$, where the $s_i$ are chosen uniformly from $\mathcal{S}$, such that $s = (\sum_{i=1}^n{s_i}) ({\rm mod}\ q)$.
\end{itemize}

It is then possible to show that for every monotone access structure $\mathcal{B}$, the SS scheme defined by $\$(s,\phi)$ satisfies the definition of a secure generalized SS scheme.

\begin{theorem}\label{thm:1}
	Let  $\mathcal{B}$ be a monotone access structure on a set $\mathcal{P}$, 
	$\phi \in \mathcal{F(B)}$ such that it is represented in $j$-notaition and contains only operators $\land $ and $\lor $, 
	and let $s$ be a secret from $\mathbb{Z}_q$. The SS scheme determined by $\$(s, \phi)$ is a secure generalized SS scheme for $\mathcal{B}$.
\end{theorem}

Finally, we note that it is shown in~\cite{Benaloh1990} that there are access structures, which cannot be realized without giving multiple (or extra large) shares to some trustee.

\begin{theorem}\label{thm:2}
    There exist access structures for which any generalized SS scheme must give some trustee shares which are from a domain
    larger than that of the secret. 
\end{theorem}

See~\cite{Benaloh1990} for the proofs of Theorem~\ref{thm:1} and Theorem~\ref{thm:2}.

\section{Advanced SS scheme} \label{sec:advss}

\subsection{General idea}
As it is noted in~\cite{Benaloh1990}, that we are unable to realize most monotone access structures with a standard SS scheme.   
However, one can modify the structures that can be realized efficiently, such that each trustee holds only one secret value, which we refer as a master key.
With the use of the master key, a trustee is able to calculate all required shares.

We define the scheme with the use of the roles as defined above.

Let us begin with an illustrative example.
Assume that a dealer wants to share a secret $s \in \mathbb{Z}_q$  between trustees Alice ($A$), Bob ($B$), Charlie ($C$), and David ($D$) according to the following access formula:
\begin{equation}
    ((X_{A,1}\land X_{B,1})\lor (X_{B,2}\land X_{C,1})\lor (X_{C,2}\land X_{D,1})),
\end{equation}
where $X_{p,j}$ is a Boolean variable that represents a trustee $p$ and appeared $j^{\rm th}$ time in the formula. 
Let us introduce an address for each variable as its position in the formula as follows:
\begin{equation}
    ((\stackrel{0}{X_{A,1}}\land \stackrel{1}{X_{B,1}})\lor (\stackrel{2}{X_{B,2}}\land \stackrel{3}{X_{C,1}})\lor (\stackrel{4}{X_{C,2}}\land \stackrel{5}{X_{D,1}}))
\end{equation}
Thus, $X_{A,1}.\address = 0$, $X_{B,1}.\address = 1$, $X_{B,2}.\address = 2$, and so on.

To share a secret, the dealer first gives each trustee $p \in \{A,B,C,D\}$ a value $\mk_{p}$, which is chosen uniformly at random from some domain $\mathcal{K}$.
Next we refer to $\mk_{p}$ as a master key belonging to a trustee $p$.

Let us then assume that the dealer and trustees have access to independent random oracles family $\{\RO_i : i \in \mathbb{Z}_q\}$ with an output domain in $\mathbb{Z}_q$. 
In order to generate a share that corresponds to a variable $X_{p,j}$, a trustee $p$ has to query the random oracle $\RO_{\mk_p}$ with $X_{p,j}.\address$.
For example, the shares for the defined access formula are computed in this way:
\begin{eqnarray}  
    s_{A,1}= \RO_{\mk_A}(X_{A,1}.\address) = \RO_{\mk_A}(0), \nonumber \\ 
    s_{B,1}=\RO_{\mk_B}(X_{B,1}.\address) = \RO_{\mk_B}(1),  \nonumber \\
    s_{B,2}=\RO_{\mk_B}(X_{B,2}.\address)=\RO_{\mk_B}(2), \nonumber \\
    s_{C,1}=\RO_{\mk_C}(X_{C,1}.\address)=\RO_{\mk_C}(3),  \\
    s_{C,2}=\RO_{\mk_C}(X_{C,2}.\address)=\RO_{\mk_C}(4), \nonumber \\
    s_{D,1}=\RO_{\mk_D}(X_{D,1}.\address)=\RO_{\mk_D}(5). \nonumber
\end{eqnarray}
Since each random oracle is independent, every share is a random value from $\mathbb{Z}_q$. 
As a result, a sum of shares, e.g. $s' := (s_{A,1}+s_{B,2}) ({\rm mod}\ q)$, is also a uniformly random variable from $\mathbb{Z}_q$.
To make it possible to reconstruct a secret $s$ by trustees $A$ and $B$, we add a publicly known value $y_1$ to this bracket, such that $ ( y_1 + s')  ({\rm mod}\ q) = s $.

Consequently, we modify our access formula into the following form:
\begin{equation}
    ((\stackrel{0}{X_{A,1}}\land \stackrel{1}{X_{B,1}} \land Y_1)\lor (\stackrel{2}{X_{B,2}}\land \stackrel{3}{X_{C,3}}\land Y_2)\lor (\stackrel{4}{X_{C,2}}\land \stackrel{5}{X_{D,1}} \land Y_3)),
\end{equation}
where $Y_i$ are Boolean variables that correspond to fictitious trustees, whose shares $y_i$ are considered to be publicly known to every actual trustee. The value of $y_i$ is computed in such a way that a reconstruction of secret becomes possible.
We note that $y_i$ is computed by the dealer, since he knows all the master keys.

Below we present our scheme in a more formal and efficient way.

\subsection{Formal Construction}
Let $n$ be a security parameter, $\mathbb{F}_q = \{f_k: k \in \mathcal{K} \}$  be a PRF family, where $q \geq 2^n$ and $f_k: \mathbb{Z}_q \rightarrow \mathbb{Z}_q$ with $|\mathcal{K}| = q$.
Here we chose $f_k: \mathcal{D} \rightarrow \mathcal{E}$ with $\mathcal{D} = \mathbb{Z}_q$, but one
can choose another domain. Note that $\mathcal{E}= \mathbb{Z}_q$, so we are able to sum the shares in $\mathbb{Z}_q$.
Let $H_q$ be a family of all functions $\mathbb{Z}_q\rightarrow\mathbb{Z}_q$. 
Let $l={\rm poly}(n)$ be the maximum size of monotone formula that we can use efficiently and let $l':=l/2$. 
Hereby the size of the monotone formula is the number of times that variables occur in the formula.

The roles for the scheme (\emph{dealer}, \emph{trustees}, and \emph{party}) are defined in the previous subsection. 

First, we define a modifying function $g_s(\phi)$, where $\phi$ is an access formula, whose size is less than $l'+1$ and it is written in $j$-notation, and $s \in \mathbb{Z}_q$.
Let $X_{p,j}$ be a variable, which represents a trustee $p$ and it is appeared $j^{\rm th}$ time in the formula. 
Let $X_{p,j}.\address$ represents the position of the variable in $\phi$.
Let $\mk_p \in \mathcal{K}$ be the value of $p$'s master key. 
We denote $Y_i$ as a variable that corresponds to a fictitious trustee and $y_i$ as the value of his share. 
We use $\phi_i$ as subformula. 
Since every formula can be written in the following form:
\begin{equation}
    \circ(\phi_1, \phi_2, \ldots , \phi_j, X_{p_1,k_1}, X_{p_2,k_2},\ldots X_{p_t,k_t}), 
\end{equation}
where $\circ$ stands for either $\land$ or $\lor$.

Let is introduce a global counter $\alpha$, which is initialized with $1$.
There are three separate cases to look at:
\begin{itemize}
    \item $g_s(X_{p_1,k_1} \land \ldots \land X_{p_t,k_t}) = (X_{p_1,k_1} 
    \land \ldots  \land X_{p_t,k_t} \land Y_\alpha)$,
     \newline where $t\geq 1$ and $y_i = s - f_{\mk_{p_1}}(X_{p_1,k_1}.\address)
     - \ldots - f_{\mk_{p_t}}(X_{p_t,k_t}.\address) ({\rm mod}\ q)$, and the counter
      $\alpha$ is incremented $\alpha := \alpha +1$.

    \item $g_s(X_{p_1,k_1} \land \ldots \land X_{p_t,k_t} \land \phi_1 \land 
    \ldots \land \phi_j) = (X_{p_1,k_1} \land \ldots \land X_{p_t,k_t} \land 
    g_{s_1}(\phi_1) \land \ldots \land g_{s_j}(\phi_j))$, where $ j \geq 1$, $
    \phi_i = \lor(\cdot)$ with at least one operator, $s_i \stackrel{\$}{\leftarrow}
    \mathbb{Z}_q$ for $i\in \{1,\ldots, j-1\}$ and $s_j := s - f_{\mk_{p_1}}(X_{p_1,k_1}.\address) - \ldots - f_{\mk_{p_t}}(X_{p_t,k_t}.\address) - s_1 - \ldots - s_{j-1} \ ({\rm mod}~q)$.
    \item $g_s(X_{p_1,k_1} \lor \ldots \lor X_{p_t,k_t} \lor \phi_1 \lor \ldots \lor \phi_j)= (g_s(X_{p_1,k_1})  \lor \ldots \lor g_s(X_{p_t,k_t}) \lor g_s(\phi_1) \lor g_s(\phi_2) \lor \ldots \lor g_s(\phi_j))$.
\end{itemize}

Let us clarify that the address of a variable is the number of the position of that variable in the formula $\phi$ (conventionally, we count from left to right). 

Now we can describe our \emph{advanced SS scheme}.
To share a secret the dealer should follow these steps:
\begin{enumerate}
    \item Choose a secret $s \stackrel{\$}{\leftarrow} \mathbb{Z}_q$.
    \item Choose a master key for each trustee in the union $\mathcal{P}$ uniformly at random from $\mathcal{K}$ (for each $p \in \mathcal{P}: \mk_p \stackrel{\$}{\leftarrow } \mathcal{K}$).
    \item Choose a monotone formula $\phi$ of size less or equal to $l'$, which represents an access structure.
    \item Evaluate $\phi' =g_s(\phi)$.
    \item Publish $\phi'$, so that the values $y_i$ are available for everyone.
\end{enumerate}
To reconstruct a secret a party should follow these steps:
\begin{enumerate}
    \item Each trustee $p$ in the party has to evaluate their shares:
     \newline  $s_{p,j} = f_{{\rm mk}_p}(X_{p,j}.\address)$.
    \item Using the corresponding shares and public values $y_i$, a verified party can calculate the secret $s$ according to the way it is shared.
\end{enumerate}

\begin{defn} Given a set $\mathcal{P}$ and a monotone access structure $\mathcal{B}$ on $\mathcal{P}$, an \textit{advanced SS scheme} for $\mathcal{B}$ is a method of dividing a secret $s$ 
    into shares $s_{p,j}$ such that the following statements hold true:
    \begin{itemize}
        \item When $B \in \mathcal{B}$, the secret $s$ can be reconstructed from the shares $\underset{p\in B}{\cup} \underset{j}{\cup} s_{p,j}$ and public values $y_1, \ldots, y_t$.
        \item When $B \notin \mathcal{B}$, the secret $s$ can be reconstructed only with a negligible probability from the shares $\underset{p\in B}{\cup} \underset{j}{\cup} s_{p,j}$ and public values $y_1, \ldots, y_t$.
    \end{itemize}
\end{defn}

\subsection{Security proof}
Here we introduce a notion of the security model that is used for our scheme, which is similar to the Selective-Id model~\cite{Canetti2003,Canetti2004,Boneh2004}.

\begin{defn}[Selective-Id model for advanced SS scheme]
    The following procedure is called Selective-Id model for advanced SS scheme.
    \begin{description}[align=left]
        \item [Init:] The adversary chooses an access structure $\mathcal{B}$ with a corresponding formula $\phi$ and gives it to the challenger.
        \item [Phase 1:] The adversary declares the set of trustees $\gamma$, which does not satisfy the formula $\phi$ and obtains master keys of trustees from $\gamma$ from the challenger.
        \item [Challenge:] The adversary submits two secrets $s_0$ and $s_1$. The challenger flips a fair coin $b$ and shares the secret $s_b$.
        \item [Phase 2:]The challenger gives to the adversary $\phi' = g_{s_b}(\phi)$ and corresponding values $y_1, \ldots, y_j$.
        \item [Guess:] The adversary outputs a guess $b'$ of $b$.
    \end{description} 
    The advantage of an adversary in this game is defined as $|\Pr[b'=b] - \frac{1}{2}|$.
\end{defn}
\begin{theorem}\label{thm:3}
    Consider the advanced SS scheme for a set of parties $\mathcal{P}$ based on  PRF family $\mathbb{F}_q =\{f_k:\mathbb{Z}_q \rightarrow \mathbb{Z}_q \}_{k\in \mathcal{K}}$ with $|\mathcal{K}| = q$.
    The advantage $\varepsilon'$ in the Selective-Id model of any classical adversary $\mathcal{A}$ that runs in time $\xi'$ satisfies the inequality
        $\varepsilon' \leq {\rm InSec^{PRF}}(\mathbb{F}_q, \xi)  \cdot |\mathcal{P}|$,
        where $\xi'\approx  \xi$ assuming that time needed for sampling no more than
         $|\mathcal{P}| +l'$ random variables is negligible, 
         where $l'$ is the maximum size of the formula which can be efficiently processed by the advanced SS scheme.
\end{theorem}

\begin{proof}

First of all, one can easily notice that the reconstruction of the secret happens the same way as in the standard SS scheme. 
We also note that if $B \notin \mathcal{B}$ (i.e. $B$ does not satisfy the formula $\phi$), then $ B\cup(\underset{i}{\cup} Y_i) $ does not satisfy the access structure defined by $\phi' = g_s(\phi)$ as $X_{p,j}\land 1 = X_{p,j}$. 
    
Consider, a modification of the advanced SS scheme (modified advanced SS scheme), where PRF family $\mathbb{F}_q$ is replaced with a set of random oracles. 
One can see that this scheme is exactly the standard SS scheme based on formula $\phi'=g_s(\phi)$. 
So there is no chance for an adversary to compute the secret, which possesses the shares from $B \notin \mathcal{B}$.

Now suppose that there exists a probabilistic polynomial time adversary $\Av$, which has an advantage $\varepsilon'$ in Selective-Id model for advanced SS scheme.
Without loss of generality, we assume that it's probability of guessing a correct value is $\Pr[b'=b]=1/2 + \varepsilon'$.
Then we show that it is possible to distinguish the PRF family $\mathbb{F}_q$ from truly random function family with a probability at least $\varepsilon'/|\mathcal{P}|$. 
To show this we construct an oracle machine $\mathcal{M}^{\Av}$ that has an advantage $\varepsilon'$ in $|\mathcal{P}|$-pseudorandom function family game (see Algorithm~\ref{alg1}).
Let us calculate the probabilities to obtain $v'=0$ and $v'=1$ ($v'$ is defined in Algorithm~\ref{alg1}).

\begin{algorithm}\label{alg1}
        \DontPrintSemicolon
        \SetKwInOut{Input}{Input}\SetKwInOut{Output}{Output}
        \Input{Security parameter $n$, function family $\mathbb{F}_q$, $|\mathcal{P}|$-pseudorandom function family challenge $\Omega = \{\omega_{p_1}, \ldots, \omega_{p_N}\}$, where $\{p_1, \ldots, p_{N}\} = \mathcal{P}$.}
        \Output{A guess $v'$.}
        The adversary $\mathcal{A}$ declares an access structure, a corresponding formula $\phi$, and a set of trustees $\gamma$, which does not satisfy the formula $\phi$.\;
        $\mathcal{A}$ queries the master keys of trustees from $\gamma$.\;
        Generate a master key uniformly at random for each trustee in $\gamma$ and response to the adversary with those keys.\;
        The adversary submits two secrets $s_0$ and $s_1$.\;
        Flip a fair coin $b$ and share the secret $s_b$ according to the \textit{advanced SS scheme}, but instead of generating master keys for trustees in $\mathcal{P} \setminus \gamma $ and calculating the shares with $f_k \in \mathbb{F}_q$, 
        use an oracle $\omega_p$ from $\Omega$ for trustee $p \in \mathcal{P} \setminus \gamma$  and calculate the shares as $s_{p.j}=\omega_p(X_{p.j}.\address)$. We call this modification $g'_s(\phi)$.\;
        Give to the adversary $\phi' = g'_{s_b}(\phi)$ and corresponding values $y_1, \ldots, y_j$.\;
        The adversary outputs a guess $b'$ of $b$.\;
        \eIf{$b' = b$}{return $v'=1$}{return $v'=0$}
        \caption{$\mathcal{M}^{\Av}$}
\end{algorithm}

Suppose that the challenge $\Omega$ is initialised with functions from the family $\mathbb{F}_q$.
In this case, the situation for the adversary is completely the same as in the case of the (non-modified) advanced SS scheme. 
Therefore, the adversary correctly guesses the value $b'$ with an advantage $\varepsilon'$ or what is the same with probability $\frac{1}{2}+\varepsilon'$. 
    
If the challenge $\Omega$ is initialized with functions from the family $H_q$, then the shares of the trustees from $\mathcal{P} \setminus \gamma$ are chosen uniformly at random. 
And the situation is the same as in the \textit{standard SS scheme}. 
Since $\gamma$ does not satisfy the formula, it is required to obtain at least one more share to get the secret, but all the remaining shares are chosen uniformly at random. 
Therefore, according to the Theorem~\ref{thm:1} the adversary has no information about the secret in this situation. 
Thus, in this case the adversary can only randomly guess the value $b$, so $b'$ is correctly guessed with a probability $\frac{1}{2}$.
    
Let $v=0$ corresponds to the challenge $\Omega$ initialized with functions from the family $H_q$ and $v=1$ to the challenge $\Omega$ initialized with functions from the pseudorandom function family.
Then the overall advantage in the $|\mathcal{P}|$-pseudorandom game is $|\Pr[v'=1|v=0] - \Pr[v'=1|v=1]| = | \frac{1}{2} - (\frac{1}{2} + \varepsilon') | = \varepsilon'$.

By the hybrid argument~\cite{Goldreich2001} we can distinguish a pseudorandom function family with probability $\varepsilon'/(|\mathcal{P}|)$.
In order to apply the hybrid argument consider two distributions, 
\begin{equation}
	D_1 = \{D_{1.i}=f_k: k \stackrel{\$}{\leftarrow}\mathcal{K} ,\ 1\leq i\leq |\mathcal{P}|, f_k\in\mathbb{F}_{q}\},
\end{equation}
and
\begin{equation}
	D_2 = \{D_{2.i}=h \stackrel{\$}{\leftarrow}H_q : 1\leq i\leq |\mathcal{P}|\}. 
\end{equation}
Define a sequence of hybrid distributions $D_1 = T_0,\ T_1, \ldots,\ T_{|\mathcal{P}|} = D_2$, where
$T_i = \{T_{i.j} = h \stackrel{\$}{\leftarrow}H_q: 1\leq j \leq i\} \cup \{T_{i.j}=f_k:k \stackrel{\$}{\leftarrow}\mathcal{K} , \ i < j \leq |\mathcal{P}|, f_k\in\mathbb{F}_{q}\}.$ So we have 
$\Adv_{D_1,D_2}(\mathcal{M}^{\Av})=\varepsilon'$. 
Let us remind that
\begin{multline}
	\Adv_{T_i,T_{i+1}}(\mathcal{M}^{\Av}) = |\Pr[x\stackrel{\$}{\leftarrow} T_i: \mathcal{M}^{\Av}(x)=1] -\\- \Pr[x\stackrel{\$}{\leftarrow} T_{i+1}: \mathcal{M}^{\Av}(x)=1]|
\end{multline}
By the triangle inequality, it is clear that 
\begin{equation*}
	\Adv_{D_1,D_2}(\mathcal{M}^{ {\mathcal{A}}}) \leq \sum_{i=0}^{|\mathcal{P}|-1} \Adv_{T_i,T_{i+1}}(\mathcal{M}^{ {\mathcal{A}}})
\end{equation*}
Thus, there exists some $\eta$, such that $0 \leq \eta < |\mathcal{P}|$ and
\begin{equation}
	\Adv_{T_{\eta},T_{\eta+1}}(\mathcal{M}^{ {\mathcal{A}}}) \geq \Adv_{D_1,D_2}(\mathcal{M}^{ {\mathcal{A}}})/|\mathcal{P}| = \varepsilon'/ |\mathcal{P}|.
\end{equation}
Suppose that we have a sample $\omega \stackrel{\$}{\leftarrow} \mathbb{F}_q$ or $\omega \stackrel{\$}{\leftarrow} H_q$. 
Let us construct a distribution $T'=\{T_i \stackrel{\$}{\leftarrow} H_q: 1\leq i \leq \eta\}\cup \{T_{\eta+1} = \omega\} \cup \{T_i \leftarrow f_k: k\stackrel{\$}{\leftarrow} \mathcal{K}, \ \eta+1<i\leq|\mathcal{P}|, f_k\in\mathbb{F}_{q}\}$.
If $\omega \stackrel{\$}{\leftarrow} \mathbb{F}_q$ then $T'$ is distributed the same as $T_{\eta}$, otherwise it is distributed as $T_{\eta+1}$. 
Thus, we can distinguish samples from $\mathbb{F}_q$ and $H_q$ with probability $\varepsilon'/|\mathcal{P}|$.

Finally, we obtain: $\varepsilon' \leq {\rm InSec^{PRF}}(\mathbb{F}_q, \xi)  \cdot |\mathcal{P}|$,
where $\xi$ is a total time of running $\mathcal{M^A}$ plus initialization of an appropriate hybrid. Neglegting the time needed
for preparing data for $\mathcal{A}$ and the hybrid $T'$ we obtain $\xi \approx \xi'$.
\end{proof}

\section{Advanced ABE Scheme} \label{sec:advabe}
\subsection{Formal Construction}
Consider a group of users, where each user posses a list of attributes.
Let $\mathcal{P}$ be a set of all existing attributes.
Let us call a community a subgroup of users, who possess a particular attribute $p\in\mathcal{P}$. 
In what follows, we refer to the community $p$ as a subgroup of users that possess an attribute $p$.
We note that a user can belong to several communities if he has more than one attribute.

Let $n\in \mathbb{N}$ be a security parameter, $G$ be a multiplicative group of a prime order $q$, where $2^n < q < 2^{n+1}$ in which DDH assumption is considered to be true, 
$g$ is a generator of that group, 
$H_q$ is a family of all functions $\mathbb{Z}_q\rightarrow\mathbb{Z}_q$,
and $\mathbb{F}_q = \{f_k:\mathbb{Z}_q \rightarrow \mathbb{Z}_q\}_{k\in G}$ is a PRF family.
We construct the advanced ABE scheme based on the advanced SS scheme in the following form.

\begin{description}[align = left]
	\item [Setup:] Each community $p$ in the universe $\mathcal{P}$ generates their secret key ${\rm sk}_p \stackrel{\$}{\leftarrow}\mathbb{Z}_q$ and a correspong public key ${\rm pk}=g^{{\rm sk}_p}$. 
	Then the public key is shared among the whole group of users.
	So that the set of public keys $ PK =\{{\rm pk}_p=g^{{\rm sk}_p}: p \in \mathcal{P}\}$ is assumed to be known to every user in the group.
	\item [Encryption $(M,PK,\phi, \mathbb{F}_q)$:] 
	To encrypt a message $M \in \mathbb{Z}_q$ under public keys $ PK$ and formula $\phi$, which represents some monotone access structure, one generates $s \stackrel{\$}{\leftarrow} \mathbb{Z}_q, e \stackrel{\$}{\leftarrow} \mathbb{Z}_q$
	and computes the ciphertext in the following form $E=\{E'=M+ s ({\rm mod}\ q), g^e, \phi'=g_s(\phi), y_1, \ldots, y_t\}$, 
	where $g_s(\cdot), y_1, \ldots, y_t$ come from the advanced SS scheme based on  PRF family $\mathbb{F}_q$ and the corresponding master keys are calculated as ${\rm mk}_p = g^{{\rm sk}_p \cdot e}$.
	\item[Decryption $(E, SK, Attr, \mathbb{F}_q)$:] , where $SK$ is a set of all secret keys known to a concrete user and $Attr$ is a set of attributes he posseses.
	For each ${\rm sk}_p \in SK$, a user calculates the master key ${\rm mk}_p = (g^e)^{{\rm sk}_p}$.
	Then if $Attr$ satisfies the access structure, then the secret $s$ can be reconstructed using $MK = \{{\rm mk}_p\}$, $\phi'$ and $y_1, \ldots, y_t$.
	The message is obtained from $E'$ as $M=E'-s ({\rm mod}\ q)$.
\end{description}
\subsection{Security proof}
In order to provide a formal security analysis of the advanced ABE scheme, we introduce the following definition.

\begin{defn}[Attribute-based Selective-Set model]
    The following \newline procedure is called attribute-based Selective-Set model:
    \begin{description}
        \item [Init:] The adversary chooses an access structure and a corresponding formula $\phi$ and sends $\phi$  to the challenger.
        \item [Phase 1:] The adversary declares the set of communities $\gamma$, which does not satisfy the formula $\phi$ and obtains secret keys of communities from $\gamma$ from the challenger.
        \item [Challenge:] The adversary submits two secrets $s_0$ and $s_1$. The challenger flips a fair coin $b$ and encrypts $m \stackrel{\$}{\leftarrow} \mathbb{Z}_q$: $E'=m + s_b ({\rm mod}\ q)$. 
        \item [Phase 2:] The challenger gives to the adversary public keys of all communities and $E$, which is a ciphertext of $m$ generated according to the advanced ABE scheme.
        \item [Guess:] The adversary outputs a guess $b'$ of $b$.
    \end{description}
    The advantage of an adversary in this game is defined as $|\Pr[b'=b] - \frac{1}{2}|$.
\end{defn}

Below we prove that the security of our scheme in the attribute-based Selective-Set model reduces to the hardness of the DDH challenge and pseudorandomness of the function family.

\begin{theorem}\label{thm:4}
    Consider the advanced ABE scheme based on an  PRF family $\mathbb{F}_q$ and set of communities $\mathcal{P}$. The advantage $\varepsilon'$ in the the Attribute-based
    Selective-Set model game of any classical adversary $\mathcal{A}$ that runs in time $\xi'$ 
    satisfies the following inequality: 
    $\varepsilon' \leq {\rm InSec^{DDH}}(G, \xi) \cdot |\mathcal{P}| + \rm{InSec^{PRF}}(\mathbb{F}_q, \tilde{\xi})\cdot|\mathcal{P}|$.
    With $\xi \approx \xi' \approx \tilde{\xi}$ assuming that time required
     for sampling no more than $3|\mathcal{P}|+l'$  random variables is negligible,
      where $l'$ is the maximum size of the formula which can be efficiently processed by the advanced ABE scheme.
\end{theorem}

\begin{proof}
First, suppose that the master keys are replaced with uniformly random keys.
In this case, let us denote the advantage in breaking the modified advanced ABE protocol in the attribute-based Selective-Set model as $\widetilde{\varepsilon}$.
If $(\varepsilon' - \widetilde{\varepsilon})$ is not negligible, then we can construct a machine that breaks $|\mathcal{P}|$-DDH challenge with an advantage of at least $(\varepsilon'-\widetilde{\varepsilon})$.

We assume that $\widetilde{\varepsilon} < \varepsilon'$, since we limit the value of $\widetilde{\varepsilon}$ by the pseudoradnomnes property and if $\varepsilon'$ is less than $\widetilde{\varepsilon}$ then we can limit them both.

Let us denote a $|\mathcal{P}|$-DDH challenge $\Omega = \{w_{p_1}, \ldots, w_{p_N}\}$, 
where$N =|\mathcal{P}|$, $\{p_1,\ldots,p_N\}=\mathcal{P}$ and $w_{p_i}$ is a tuple either $(g^a, g^{b_i}, g^{a\cdot b_i})$ or $(g^a, g^{b_i}, g^{z_i})$. We use $w_{i.j}$ to denote the $j^{\rm th}$ element of the tuple.
To prove the theorem, consider the following algorithm.

\begin{algorithm}\label{alg2}
        \DontPrintSemicolon
        \SetKwInOut{Input}{Input}\SetKwInOut{Output}{Output}
        \Input{Security parameter $n$, $|\mathcal{P}|$-DDH challenge $\Omega$.}
        \Output{A guess $v'$.}
        The adversary $\mathcal{A}$ chooses an access structure and a corresponding formula $\phi$ and sends it to the challenger.\;
        $\mathcal{A}$ declares the set of communities $\gamma$, which does not satisfy the formula $\phi$, whose secret keys he wishes to get and queries them.\;
        Generate a secret key for each community in $\gamma$ and response to the adversary with those keys.\;
        The adversary submits two secrets $s_0$ and $s_1$.\;
        Flip a fair coin $b$ and encrypt a message $m \stackrel{\$}{\leftarrow}\mathbb{Z}_q$ according to the advanced ABE scheme with $s=s_b$, 
        but instead of secret keys for communities in $\mathcal{P} \setminus \gamma $ use sample $\omega_p$ from $\Omega$ for community $p \in \mathcal{P} \setminus \gamma$. 
        Take $\omega_{p.2}$ as his public key and $\omega_{p.3}$ as his master key. 
        We call this modification $g'_s(F)$.\;
        Give to the adversary $E=\{E'=m+s ({\rm mod}\ q),\omega_{1,1},\phi' = g'_{s}(\phi),y_1, \ldots, y_j\}$.\;
        The adversary outputs a guess $b'$ of $b$.\;
        \eIf{$b' = b$}{return $v'=1$}{return $v'=0$}
        \caption{$\mathcal{M}^{\Av}$}
\end{algorithm}
    
    If $\omega_{p.3}$ is sampled uniformly at random ($v=0$), then the master keys are chosen uniformly at random. 
    Hence the adversary has no information about the master keys he did not query. 
    Remind that we denote the advantage of the adversary in this situation as $\widetilde{\varepsilon}$. 
    Otherwise $(v=1)$ the situation is the same as in the original ABE protocol.
    Thus, we have the overall advantage in the 
    $|\mathcal{P}|$-DDH game as follows: 
    
    \begin{multline}
        {\rm InSec^{|\mathcal{P}|-DDH}}(G, \xi_{\mathcal{P}}) \geq |\Pr[v'=1|v=0] - \Pr[v'=1|v=1]| =\\
        =  |(\frac{1}{2}+\widetilde{\varepsilon}) - (\frac{1}{2} +\varepsilon' ) | = \varepsilon'-\widetilde{\varepsilon},
    \end{multline}
    where $\xi_{\mathcal{P}}$ is a running time of Algorithm~\ref{alg2}. Neglegting the time for preparing data for $\mathcal{A}$
    we obtain $\xi_{\mathcal{P}} \approx \xi'$.

    In analogy to the proof of Theorem~\ref{thm:3}, one can see that due to the hybrid argument
    \begin{equation*}
        {\rm InSec^{DDH}}(G, \xi) \geq (\varepsilon' -\widetilde{\varepsilon})/|\mathcal{P}|,
    \end{equation*}
    where $\xi \approx \xi_{\mathcal{P}}$ neglegting the time, needed to prepare an appropriate hybrid.

    Finally, we limit the value of $\widetilde{\varepsilon}$. 
    Due to the fact the master keys are chosen uniformly at random, the security of such a scheme reduces to the security of the advanced SS scheme straightforwardly. 
    Therefore, according to Theorem~\ref{thm:3}: $\widetilde{\varepsilon} \leq \rm{InSec^{PRF}(\mathbb{F}_q,\tilde{\xi})}\cdot|\mathcal{P}|$, with $\tilde{\xi} \approx \xi_{\mathcal{P}}$.

    Thus, we arrive to the final result:
    \begin{equation*}
        \varepsilon' \leq {\rm InSec^{DDH}}(G,\xi) \cdot |\mathcal{P}| + \rm{InSec^{PRF}(\mathbb{F}_q,\tilde{\xi})}\cdot|\mathcal{P}|,
    \end{equation*}
with $\xi' \approx \xi \approx \tilde{\xi}$.

\end{proof}
\section{Efficiency estimation for advanced ABE scheme }\label{sec:efficiency}

Here we analyze the efficiency of the proposed advanced ABE scheme in terms of sizes of ciphertext, public parameters, and private key,
and the computation time for decryption and encryption.

Consider a ciphertext $E=\{E' = m + s ({\rm mod}\ q), g^e, \phi'=g_s(\phi), y_1, \ldots, y_j\}$ and a plaintext $ PT=\{m, \phi\}$. 
We note that it is required to publish the rules of the access structure, hence we assume that the plaintext is accomplished by the formula $\phi$.
One can see that $\phi'$ is no more than twice bigger than $\phi$.
This is due to the fact that the number of additional Boolean variables corresponded to fictitious trustees (communities)
can not exceed the number of Boolean variables corresponded to the actual trustees (communities).
We then note that $\phi$ and $\phi'$ make a major contribution into the size of $ PT$ and $E$.
Hence, the overhead of the ciphertext compared to plaintext is of the size linear in the size of the formula $\phi$.

The public parameters of the system are of size linear in the number of existing attributes.
The private key of the community consists of a single value from $\mathbb{Z}_q$.

The encryption procedure generates two random values, performs one addition in $\mathbb{Z}_q$ and one exponentiation in the group $G$, $l$ calls to functions from $\mathbb{F}_q$, where $l$ denotes the size of the formula $\phi$.
The modification of the formula $\phi$ into $\phi'$ is performed in the linear time with the usage of syntax tree. 

Thus, the amount of the communities in the scheme is $|\mathcal{P}|$.
The decryption procedure needs to perform at most $|\mathcal{P}|$ exponentiations, $l'$  sums and pseudorandom function calls, where $l'$ is the size of formula $\phi'$. 
Finally, one subtraction is required.

\section{Conclusion} \label{sec:conc}
Here we summarize the main results of our work.
First, we have presented the modification of the SS scheme, which allows a user to store only one value to calculate the corresponding shares. 
Based on this modification, we have proposed the advanced ABE protocol. 
We have provided the security and efficiency analysis of the proposed scheme.

One of the most significant impacts of this paper is rejection of bilinear mappings, which evidently increases the
efficiency of the proposed scheme and allows to dinamicaly add new attributes.

One can see that the proposed ABE scheme is not collusion resistant as well as some other ABE schemes (e.g. see~\cite{Kapaida2007}).
We note, that all known collusion resistant schemes are based on using of trusted centers which are absent in our scheme. 

There are several ways to improve the proposed scheme. 
The first one is based on adding new logical elements, e.g. threshold, so that the formula $\phi$ can be constructed more efficiently. 
The second question is related to modification of this protocol with respect to the use of other key exchange schemes. 

\section*{Acknowledgments} 
This work is supported by Russian Foundation for Basic Research (18-37-20033). 
A.A.C. is supported by Russian Science Foundation (17-11-01377).

\end{document}